\documentclass[letterpaper, 10 pt, journal, twoside]{IEEEtran} 
\IEEEoverridecommandlockouts                        

\pdfoutput=1

\usepackage{times}    
\usepackage{amsmath}  
\usepackage{amssymb}  
\usepackage{amsthm}
\usepackage{amsfonts,epsfig,epstopdf}
\usepackage{mathrsfs,bbm}
\newtheorem{definition}{Definition}
\newtheorem{theorem}{Theorem}
\newtheorem{lemma}{Lemma}

\newtheorem{proposition}{Proposition}

\newcommand{\ie}{{i.e.,}}

\def\1{{\bf 1}_N}
\def\0{{\bf 0}}

\title{ Allocating marketing resources over social networks: A long-term analysis\thanks{This work was partially supported by  INS2I CNRS under the 80'PRIME call.}}

\author{Vineeth S. Varma, Samson Lasaulce, Julien Mounthanyvong and Irinel-Constantin Mor\u{a}rescu\thanks{V. S. Varma and I-C. Mor\u{a}rescu are with the Universit\'e de Lorraine, CNRS, CRAN, F-54000 Nancy, France, {\small \tt constantin.morarescu@univ-lorraine.fr.}}
\thanks{S. Lasaulce and J. Mounthanyvong are with the Laboratoire des Signaux et Systemes (L2S, CNRS-CentraleSupelec-Univ. Paris Sud), Gif-sur-Yvette, France.} }%

\begin{document}

\maketitle
\thispagestyle{empty}

\begin{abstract} In this paper, we consider a network of consumers who are under the combined influence of their neighbors and external influencing entities (the marketers). The consumers' opinion follows a hybrid dynamics whose opinion jumps are due to the marketing campaigns. By using the relevant static game model proposed recently in \cite{varma2018marketing}, we prove that although the marketers are in competition and therefore create tension in the network, the network reaches a consensus. Exploiting this key result, we propose a coopetition marketing strategy which combines the one-shot Nash equilibrium actions and a policy of no advertising. Under reasonable sufficient conditions, it is proved that the proposed coopetition strategy profile Pareto-dominates the one-shot Nash equilibrium strategy. This is a very encouraging result to tackle the much more challenging problem of designing Pareto-optimal and equilibrium strategies for the considered dynamical marketing game. 
\end{abstract}


%

\section{INTRODUCTION}
\label{sec:introduction}

\IEEEPARstart{I}{n} many domains such as in economics and politics, people (e.g., consumers or voters) are both influenced by their acquaintances, friends, or relatives and by external entities (e.g., marketers or candidates); these influencers are called in a generic manner \textit{marketers}. These external entities are currently better realizing the potential of acquiring and exploiting some knowledge about the corresponding dynamics of a digital social network to design good strategies. Targeted and viral marketing constitute good examples illustrating this tendency \cite{tuten2017social}. To provide a specific example, quite recently, some firms have been starting to remunerate popular bloggers or YouTubers to promote some goods in their videos. The main purpose of the present paper is precisely to study the evolution of people opinion when they are under the combined influence of their "neighbors" (who may have different degrees of influence) and marketers (who typically have diverging interests). Whereas opinion dynamics (OD) has been attracting a lot of attention from researchers, in the control community, in particular, the problem of controlling opinion dynamics has been left almost unexplored. Additionally, if one considers the problem in presence of multiple controllers instead of one, then only a couple of formal works seem to be available. 

Among relevant works on controlled OD, we find \cite{Camponigro, Dietrich} in which the authors look at the role of controlling (from a single controller) a small number of agents of the network to enforce consensus. We also find recent attempts to control the discrete-time dynamics of opinions such that as many agents as possible reach a certain set after a finite number of influence instances \cite{Hegselmann}. The classical literature on non-cooperative games between marketers assumes a homogeneous population of consumers \cite{friedman1958game, butters1977equilibrium,esmaeili2009game}. For the scenario which directly concerns the present work, namely the scenario that involves multiple controllers or marketers influencing consumers over social media, the closest works are given by  \cite{masucci2014strategic} and \cite{varma2018marketing}. In \cite{masucci2014strategic}, the authors consider multiple influential entities competing to control the opinion of consumers under a game-theoretical setting.  This work assumes an undirected graph and a (specific) voter model for OD resulting in strategies that are independent of the node centrality (\ie \ the agent influence power). On the other hand, in \cite{varma2018marketing}, the authors use the node centrality to define the agent influence power and show how the marketers can exploit this quantity to allocate their marketing budget over the agents, and therefore "optimize" their return of investment in terms of market share. The authors then use a static or one-shot game model and conduct the corresponding Nash equilibrium (NE) analysis. The obtained results clearly show the benefit of designing target marketing strategies by using the available knowledge about the graph of the network of agents.  However, this interesting analysis is incomplete as it is assumed that each marketer makes decisions independently from campaign to campaign. Moreover, when the marketers implement the derived one-shot NE strategies, one does not know about the long-term behavior of the marketers. Remarkably, a long-term analysis, as conducted in the present letter, reveals that the marketers may have an interest in stopping to invest and therefore influence the consumers and accept to operate at a network equilibrium point in terms of market shares. 


{\bf Notation.} Let $\mathbb{R}_{\geq 0}:= [0,\infty)$ denote the set of non-negative real numbers. If $f(t)$ is a lower semi-continuous function at $t_0$, we use the notation $f(t^+_0)$ to imply $f(t^+_0):= \lim_{t \to t_0, t>t_0} f(t)$. Since we are concerned with a duopoly in this work, for ease of exposition, we will denote by $-i$ when $i \in \{1,2\}$ is a player index, to refer to the index of the other player, i.e. $-i := 1+i \mod 2$.

\section{PROPOSED PROBLEM FORMULATION}
\label{sec:pb-formulation}

We assume the presence of two marketers who want to capture agents (who will also be called consumers) over a common market. The set of consumers is denoted by $\mathcal{N}= \{1,2,\dots,N \}$; these consumers are connected through a social network. The opinion of Consumer $n \in \mathcal{N}$ at time $t  \geq 0$ is represented by the scalar $x_n(t) \in (0,1)$, with $x(t)\in \mathcal{X}_0$ and $\mathcal{X}_0:=(0,1)^N$. The vector $x(t) = (x_1(t),x_2(t),\dots,x_N(t))^\top$ is called the state of the network at time $t$. In the absence of the marketers, the opinions evolve based on a consensus model with Laplacian matrix $\mathbf{L}$ over a graph $(\mathcal{N},\mathcal{E})$. 
At given time instances which are referred to as marketing campaigns, consumers undergo the influence of the marketers. The set of marketing campaign instances is denoted by $\mathcal{T}= \{t_1,t_2,\dots,t_K\}$, $K$ being the number of campaigns; the set of campaign indices is denoted by $\mathcal{K}:=\{1,2,\dots,K\}$. The campaign duration for Stage $k  \in \{1,\dots,K-1\}$ is given by $T_k \in (0,t_{k+1}-t_k]$ and $T_K>0$ for Stage $K$. At each time instant $t_k \in \mathcal{T}$, Marketer $i \in \{1,2\}$ invests according to the (action) vector $a_i(k) = (a_{i,1}(k), \dots, a_{i,N}(k))^\top \in \mathcal{A}_i$ where the corresponding action space for Marketer $i$ is defined as: $\mathcal{A}_i = \left\{a_i \in \mathbb{R}^N: a_{i,n} \geq 0, \sum_{n=1}^N a_{i,n} \leq B_i \right\}$, $B_i > 0 $ being the available budget for Marketer $i$. As a result of the marketing campaigns, the OD for the consumers is assumed to obey the following hybrid model

\begin{equation}
\left\{\begin{array}{llll}
\dot{x}(t)&=& -\mathbf{L} x(t)& \forall t \in \mathbb{R} \setminus \mathcal{T} \\
x(t_k^+) &= &\Phi(x(t_k), a_{1}(k),a_{2}(k))& \forall t_k \in \mathcal{T} , k \in \mathcal{K}, 
\end{array}\right.
\label{eq:opdyn}
\end{equation}

where $\Phi(x(t_k), a_{1}(k),a_{2}(k))  = \left(\phi(x_1(t_k), a_{1,1}(k),a_{2,1}(k) ),..,   \phi(x_N(t_k), a_{1,N}(k),a_{2,N}(k) )   \right)^\top $ and

\begin{equation}
\phi(x_n(t_k), a_{1,n}(k),a_{2,n}(k) )    =\frac{x_{n}(t_k)+ a_{1,n}(k)}{1+a_{1,n}(k)+a_{2,n}(k)}. 
\end{equation}

The assumed jump model has been proposed in \cite{varma2018marketing} and it is justified therein by an axiomatic approach. In fact,  it can also be justified by other good arguments e.g., by probabilistic arguments \cite{martins2008continuous} or from an economic resource allocation point of view \cite{kelly1998rate}. The actions of the marketers are assumed to be driven by their utility function. The \textit{stage revenue or utility} for Marketer $i$ (that is, resulting from the current campaign) for Campaign $k \in \mathcal{K}$ is assumed to be given by the one-shot game model developed in \cite{varma2017opinion}  that is:  

\begin{equation}
\hspace{-2mm}\begin{array}{l}
u_1(x(t_k^+) , a_{1}(k), a_2(k)) = \rho(k)^\top x(t_k^+) - \lambda_1 \mathbf{1}_N^\top  a_{1}(k) \\
u_2(x(t_k^+), a_{1}(k), a_{2}(k)) = \rho(k)^\top ( \mathbf{1}_N - x(t_k^+) )  -\lambda_2  \mathbf{1}_N^\top a_{2}(k)  
\end{array}
\label{eq:uiv1}
\end{equation}

where $\mathbf{1}_N$ is the column vector of $N$ ones and $\lambda_i \geq 0$ is a parameter that represents the cost of advertising for Marketer $i$. Notice that the assumed utilities can be seen as the result from an averaging effect over the opinion. Indeed, in \cite[ Prop. 1]{varma2017opinion}, it is shown that one can write that $\int_{t_k}^{t_k+T_k} 1_N^\top x(t) \mathrm{d}t  = \rho(k)^\top x(t_k^+)$, where $\rho(k)$ depends on $T_k$ and $\mathbf{L}$. Now, if the utilities have to be related to the final opinion only, observe that  $1_N^\top x(t_k+T_k)= \rho(k)x(t_k^+)$ where $\rho(k)^\top = 1_N^\top \exp(-\mathbf{L}^\top T_k)$. These are two different situation justifying the form of utilities in \eqref{eq:uiv1} in which only the expression of $\rho$ changes. We will refer to the latter key quantity as the \textit{agent influence power} for Consumer $n$ over Campaign $k$.

In \cite{varma2018marketing}, the authors suggest that a possible strategy is that a marketer chooses, at each stage (or campaign), its (unique) NE action associated with the static game defined by $\left( \{1,2\}, \mathcal{A}_1 \times \mathcal{A}_1, \{u_i\}_{i \in \{1,2\}} \right)$. Here, to conduct a long-term analysis of the problem, we consider a setting which encompasses that model. Indeed, we consider long-term utilities which result from averaging stage utilities over the $K$ stages. To define these utilities, we first define marketing strategies. \textit{The marketing strategy} for Marketer $i$ is the sequence of functions denoted by $\sigma_i$ and defined by:

\begin{equation}
\begin{array}{ccccc}
\sigma_{i,k}  & : & \mathcal{H}_k & 
\to & \mathcal{A}_i\\
                     &    &  h(k) &  \mapsto & a_i(k)  
\end{array}
\end{equation}

where $\mathcal{H}_k   = ( \mathcal{X}_0 \times \mathcal{A}_1 \times \mathcal{A}_2)^{k-1} $ is the set of possible histories of the long-term game at stage $k$ and $h(k) =\left(x(t_1),a_1(1),a_2(1), \dots,   x(t_{k-1}),a_1(k-1),a_2(k-1) \right) $ is the history realization  at stage $k$. The \textit{long-term utility} or total net revenue for Marketer $i$ is then given by:

\begin{equation}
U_i(\sigma_1, \sigma_2) = \frac{1}{K} \sum_{k=1}^{K} u_i(x(t_k),a_{1}(k),a_2(k)).
\end{equation}

One of the goals of this paper is to design good marketing strategies whose performance are measured in terms of long-term utility. Notice that the problem under consideration is a hybrid dynamic game with causal closed-loop feedback and perfect monitoring. Both the characterization of equilibrium utilities and the determination of good strategies for such game models is known to be non-trivial. One big difference between the present letter and \cite{varma2018marketing} is as follows. For the one-shot game, expressing the best-responses is shown to be possible in the latter and thus, by intersection, the one-shot Nash equilibrium \textbf{actions} are obtained. This is not possible to do so when it comes to \textbf{strategies}, which are (possibly infinite) sequences of functions. In this paper, we make the first step into solving this problem by analyzing the long-term performance of the repeated NE strategy and by exhibiting a feasible strategy which outperforms the one-shot NE strategy.

\section{PERFORMANCE ANALYSIS OF THE ONE-SHOT GAME MARKETING STRATEGY}
\label{sec:perf-analysis}

From (\ref{eq:opdyn}), it is seen that the continuous-time component of the considered hybrid dynamical system corresponds to a consensus model. On the other hand, the jumps associated with the discrete-time part are a result of the choices made by the decision-makers (namely, the two marketers) who have non-aligned utilities. In fact, if the costs of advertising are zero, each stage game is strictly zero-sum.Therefore, in the presence of external influencers who have diverging interests and impact the dynamics with an infinite number of jumps, it is not clear whether the jumps will vanish and a consensus will emerge. Remarkably, it is possible to show that when the marketers choose the action corresponding to the NE of each stage game, the network state stabilizes to a value and this in the presence of marketers in interaction) whose expression is very simple and elegant. For this, we first make the following observations. Even if the  number of marketing campaign would be arbitrarily large, the marketers would not use all of it. This is because, by construction of the utility, strong influencing actions also involve a cost, which naturally regularizes the behavior of the marketer. In fact, it is possible to exhibit a budget threshold above which the marketers have an interest in using the extra budget. It turns out that assuming the available budget is above this threshold, it becomes possible to re-express the repeated one-shot NE strategy to prove the convergence of the network state as stated in the theorem provided further.
\begin{lemma} Let 

\begin{equation}
\overline{X}_n: = \left\{y \in \mathbb{R}:y> 1 - \frac{\lambda_1}{\rho_n} , y< \frac{\lambda_2}{\rho_n}  \right\} \label{eq:defbarX}
\end{equation}
and 

\begin{equation}
X^\dagger_n: = \left\{y \in \mathbb{R}: y> 1 -\frac{(1-\eta) \rho_n}{\lambda_1+\lambda_2}, y< \frac{\eta \rho_n}{\lambda_1+\lambda_2}  \right\},
\end{equation}

where 

\begin{equation}
\eta := \frac{\lambda_2}{\lambda_1+ \lambda_2}. \label{eq:eqxNE}
\end{equation}
Assume that the budget for each agent satisfies the following relation

\begin{equation}
B_i \geq \sum_{n\in \mathcal{N}} \max\left\{0,\sqrt{\frac{\rho_n}{\lambda_i}}-1,\frac{\rho_n}{\lambda_1+\lambda_2}-1\right\}. \label{eq:budsur}
\end{equation}

Then, for each $n \in \mathcal{V}$, the stage game NE expresses as follows.
\begin{enumerate}
\item When $\frac{\rho_n}{\lambda_1+\lambda_2} \leq 1 $:
\begin{equation}\begin{array}{l}
(a_{1,n}^\star,a_{2,n}^\star) =\\
 \left\{ \begin{array}{ll}
(0,0) & \text{if }x_n(t_k) \in \overline{X}_n \\
(0,\sqrt{\frac{\rho_n}{\lambda_2} x_n(t_k)} - 1) & \text{if }x_n(t_k) \geq \max \overline{X}_n \\
(\sqrt{\frac{\rho_n}{\lambda_1} (1-x_n(t_k)}) - 1,0) & \text{if }x_n(t_k) \leq \min \overline{X}_n.
\end{array} \right.
\end{array}
\label{eq:OSNEB}
\end{equation}

\item When $\frac{\rho_n}{\lambda_1+\lambda_2} > 1 $:

\begin{equation}\begin{array}{l}
a_{1,n}^\star = \\
\left\{ \begin{array}{ll}
\frac{\rho_n \eta }{\lambda_2+ \lambda_{1}}  - x_{n}(t_k) & \text{if }x_n(t_k) \in X_n^\dagger\\
0 & \text{if }x_n(t_k) \geq \sup X_n^\dagger \\
\sqrt{\frac{\rho_n}{\lambda_1} (1-x_n(t_k)}) - 1 & \text{if }x_n(t_k) \leq \inf X_n^\dagger
\end{array} \right. \end{array}
\label{eq:OSNE1}
\end{equation}

and

\begin{equation}\begin{array}{l}
a_{2,n}^\star = \\
 = \left\{ \begin{array}{ll}
\frac{\rho_n (1-\eta) }{\lambda_2+ \lambda_{1}}  - (1-x_{n}(t_k)) & \text{if }x_n(t_k) \in X_n^\dagger\\
0 & \text{if }x_n(t_k) \leq \inf X_n^\dagger \\
\sqrt{\frac{\rho_n}{\lambda_2} x_n(t_k)} - 1)& \text{if }x_n(t_k) \geq \sup X_n^\dagger.
\end{array} \right. \end{array}\label{eq:OSNE2}
\end{equation}

\end{enumerate} \label{prop:OSNE}
\end{lemma}
\begin{proof}
From \cite{varma2018marketing}, we know that the best response for player $i$ to action $a_{-i}$ by the other player is given by
\begin{equation}
\mathrm{BR}_i = \max\left\{0,\sqrt{\frac{\rho_n (a_{0,n;-i} + a_{-i,n} )}{\mu_{i}+ \lambda_i} } - 1 -a_{-i,n} \right\} \label{eq:br}
\end{equation}

where $\mu_{i}$ is a common constant for all $n \in \mathcal{V}$, and is a

 Lagrange multiplier which ensures that the budget constraint $\sum_{n=1}^N a_{i,n} \leq B_i$ is satisfied. Proposition 2 in \cite{varma2018marketing} shows that for each $n \in \mathcal{V}$, the NE $(a_{1,n}^\star,a_{2,n}^\star)$ is given by
\begin{itemize}
\item $(y,0)$ (or $(0,y)$) if $\exists y \in [0,\infty]$ such that \eqref{eq:br} is satisfied by one of these pairs,
\item or $(a_{1,n}^\star,a_{2,n}^\star) \in (0,\infty) \times (0,\infty)$ and is given by

\begin{equation}
a_{i,n}^\star = \left(\frac{ k_i}{k_i +k_{-i}}\right)^2 k_{-i}\rho_n - a_{0,n;i}, \label{eq:NEspe}
\end{equation}

where $k_i= \frac{1}{\lambda_i+\mu_{i}}$ and $\mu_{i}$ is a common constant for all $n \in \mathcal{V}$ such that $\sum_{n=1}^N a_{i,n} \leq B_i$.
\end{itemize}

In the absence of budget constraints, we are able to set $\mu_{i}=0$ for \eqref{eq:br} and for Case 2 of Lemma 1. We then show that the resulting actions respect the budget constraint when $B_i$ satisfies \eqref{eq:budsur}. $\sum_{n=1}^N a_{i,n} \leq B_i$. Therefore \eqref{eq:NEspe} yields

\begin{equation}
a_{i,n}^\star = \left(\frac{1}{\lambda_i +\lambda_{-i}}\right)^2 \lambda_{-i}\rho_n - a_{0,n;i}. \label{eq:NEspe2}
\end{equation}

Note that $\frac{\rho_n}{\lambda_1+\lambda_2} \leq 1 \Rightarrow  \bar{X}_n \neq \emptyset$ and $X^\dagger_n = \emptyset$ which corresponds to Case 1 of Lemma 1; while $\frac{\rho_n}{\lambda_1+\lambda_2} > 1 \Rightarrow \bar{X}_n = \emptyset$ and $X^\dagger_n \neq \emptyset$ corresponding to Case 2 of Lemma 1. In case $1$, we find that \eqref{eq:NEspe} will never have positive actions for both players simultaneously for any $x_n(t_k) \in [0,1]$ (note that  $a_{0,n;1}=x_n(t_k)$ and $a_{0,n;2}=1-x_n(t_k)$). Therefore the only possible solutions are as given in Case 1 of Lemma 1 by looking at \eqref{eq:br} with one action set to $0$. On the other hand, Case 2 of Lemma 1 is possible when $x_n(t_k)$ belongs to the open interval $X^\dagger_n$. Outside this interval, we take Case 1 of Lemma 1 again to get the final results. The largest action that can be taken under case 2 is bounded by 
\[ \sup \left\{  \frac{\rho_n \eta}{\lambda_1+\lambda_2} - x_n: x_n \in X^\dagger_n \right\} \]
which is less than $\frac{\rho_n }{\lambda_1+\lambda_2}-1$ since $\eta<1$ and $\inf X^\dagger_n= 1 -\frac{(1-\eta) \rho_n}{\lambda_1+\lambda_2} $. Under Case 1, the maximum is simply $\sqrt{\frac{\rho_n \eta}{\lambda_i} } -1$. Applying the same logic for all agents we get that the total budget is always less than $B_i$ if $B_i$ respects \eqref{eq:budsur}.
\end{proof}
  
We will refer to the strategy associated with playing the one-shot NE at every stage as $\sigma_i^{\star}$. Exploiting the above result, the following theorem can be proven.

\begin{theorem} Let $\rho_{\max} := \min_{k \in \mathcal{K}} \max_{n \in \mathcal{N}} \rho_n(k)$. Assume Marketer $i$, $i \in \{1,2\}$, implements the marketing strategy $\sigma_i^{\star}$. Assume the graph associated with the matrix $\mathbf{L}$ to be strongly connected. Then the dynamical system \eqref{eq:opdyn} has at least one (network) equilibrium $x^*$ which verifies the following:
\begin{itemize}
\item If $\frac{  \rho_{\max}  }{\lambda_1+\lambda_2} >1$, then $x^* = \eta \mathbf{1}_N$ is the unique network equilibrium.
\item If $\frac{  \rho_{\max}  }{\lambda_1+\lambda_2} \leq 1$, then any $x^* = \gamma 1_N $ is a network equilibrium, with $\gamma \in \overline{X}_{\max}$, where $\overline{X}_{\max}$ is defined by replacing $\rho_n$ with $\rho_{\max}$ in \eqref{eq:defbarX}. 
\end{itemize} 
\end{theorem}
\begin{proof} Our proof is structured as follows.  First, we show that if $\frac{\rho_{\max}}{\lambda_1+\lambda_2} \leq 1$ any $x^* \in \bar{X}_{\max}$ is a network equilibrium. Next, we show that if $\frac{\rho_{\max}}{\lambda_1+\lambda_2} > 1$, $||x(t_{k+1})- \eta \mathbf{1}_N||_\infty < ||x(t_k)- \eta \mathbf{1}_N||_\infty$ for all $x(t_k) \neq \eta \mathbf{1}_N$ implying convergence to $\eta\mathbf{1}_N$ which is the unique equilibrium.\\
Since the flow dynamics of \eqref{eq:opdyn} are basically consensus type dynamics, we know that $\gamma \mathbf{1}_N $ is a network equilibrium for any $\gamma \in \mathbb{R}$ for the part $\dot{x}=-\mathbf{L}x$.  If $\gamma \mathbf{1}_N=\Phi(\gamma \mathbf{1}_N,a_1^\star,a_2^\star)$, then we know that $\gamma \mathbf{1}_N $ is a network equilibrium. 

\textbf{Case 1: } When $\frac{\rho_{\max}}{\lambda_1+\lambda_2} > 1$: This implies that for any stage $k$, we have at least one $m$ such that $\frac{\rho_m}{\lambda_1+\lambda_2} >1$. From Lemma \ref{prop:OSNE}, we know that this implies $X^\dagger_n \neq \emptyset$. Therefore, the actions $a_{i,m}(k)$ are never simultaneously $0$. If $x_m(t_k) \in X^\dagger_m$, we have

\begin{equation}\begin{array}{ll}
x_m(t_k^+) &= \frac{x_m(t_k) + \frac{\rho_m \eta}{\lambda_2 + \lambda_1} -  x_n(t_k)  } {1 + \frac{\rho_m \eta}{\lambda_2 + \lambda_1} -  x_n(t_k) + \frac{\rho_m (1-\eta)}{\lambda_2 + \lambda_1} - (1- x_m(t_k))  }  \\[1mm]
&= \frac{ \rho_m   \eta}{ \rho_m \eta +  \rho_m(1- \eta)} = \eta
\end{array} \label{eq:etastable}
\end{equation}

using \eqref{eq:OSNE1} and \eqref{eq:OSNE2}. If $x_m(t_k) \notin X^\dagger_m$, we have exactly one of the players' actions non-zero which means that it is not a network equilibrium.

For any $n$ such that $\frac{\rho_{n}}{\lambda_1+\lambda_2} \leq 1$, we know that $a_{1,n}=a_{2,n}=0$ when $x_n(t_k) \in \overline{X}_n$ from \eqref{eq:OSNEB}. However, we can easily show that $\eta \in \overline{X}_n$. First, we see  

$$\max  \overline{X}_n= \frac{\lambda_2}{\rho_n}= \frac{\eta (\lambda_1 + \lambda_2)}{\rho_n} > \eta.$$

Similar arguments can be used to show that $\eta  \geq \min \overline{X}_n$. This implies that $x_n(t_k^+)=\eta$ if $x_n(t_k)=\eta$ for any $n$ and $k$. Therefore, the only value of $\gamma$ such that $\gamma \mathbf{1}_N=\Phi(\gamma \mathbf{1}_N,a_1^\star,a_2^\star)$ is when $\gamma=\eta$.

\textbf{Convergence of the hybrid dynamics:} Rewriting the flow dynamics, we have 

\begin{equation}
\begin{array}{l}
||x(t_{k+1})- \eta \mathbf{1}_N||_\infty  
=|| \exp(-\mathbf{L} T_k) [ x(t_{k}^+)- \eta \mathbf{1}_N ]||_\infty.\\
\end{array}
\end{equation}

When $x_m(t_k) \leq \inf X_m^\dagger <  \eta$, $a_{2,m}^\star =0$ and we use \eqref{eq:OSNE2} and \eqref{eq:OSNE1} to get

\begin{equation}
x_m(t_k^+) < 1 - \sqrt{(1-\eta) \frac{\lambda_1}{\lambda_1+\lambda_2}  } < \eta
\end{equation}

Thus, $|x_m(t_k^+) - \eta| < |x_m(t_k) - \eta| $ when $x_m(t_k) \notin X^\dagger_n$. 

On the other hand, for $n$ such that $\frac{\rho_{n}}{\lambda_1+\lambda_2} \leq 1$ we have the following. If $x_n(t_k) \leq \min \overline{X}_n$, $a_{2,n}^\star =0$ and we can use \eqref{eq:OSNE1} to solve for 

\begin{equation}\begin{array}{ll}
x_n(t_k^+) &= \frac{x_n(t_k) + \sqrt{\frac{\rho_n x_n(t_k)}{\lambda_1}} -1 } {1 + \sqrt{\frac{\rho_n x_n(t_k)}{\lambda_1} }-1  }  \\
&= 1 - \sqrt{ \frac{\lambda_1 (1- x_n(t_k)) }{\rho_n} } \leq 1 - \sqrt{ \frac{(1-\eta) \lambda_1}{\rho_n}}\\
\end{array} \label{eq:ineqxplus}
\end{equation}

However, $\frac{\lambda_1}{\rho_n} \geq 1- \eta$ and therefore $x_n(t_k^+) \leq \eta$. By similar calculations we can show that $x_n(t_k^+) \leq \eta$ when $x_n(t_k) \geq \max \overline{X}_n $. This implies that $|x_n(t_k^+) - \eta| < |x_n(t_k) - \eta| $ when $x_n(t_k) \notin X^\dagger_n$. 

Coupled with \eqref{eq:etastable}, we have shown that $|x_m(t_{k}^+)- \eta| < |x_m(t_{k})- \eta|$ unless $x_m(t_m)=\eta$. Therefore, unless $x_n(t_k)=\eta$ for all $n$, we have

\begin{equation}\begin{array}{r}
|| \exp(-\mathbf{L} T_k) [ x(t_{k}^+)- \eta \mathbf{1}_N ]||_\infty \leq  ||x(t_{k}^+)- \eta \mathbf{1}_N|^2 ||_\infty \\
< || x(t_{k})- \eta \mathbf{1}_N||_\infty
\end{array}
\end{equation}

On the other hand, if $x_m(t_m)=\eta$, either $x(t_{k})=\eta \mathbf{1}_N$ or we have at least one $n$ such that $x_n(t_k) \neq \eta$. In the first case, the network is already at equilibrium and in the other case, we have

\begin{equation}\begin{array}{r}
|| \exp(-\mathbf{L} T_k) [ x(t_{k}^+)- \eta \mathbf{1}_N ]||_\infty < ||x(t_{k}^+)- \eta \mathbf{1}_N|^2 ||_\infty \\
\end{array}
\end{equation}

as $\exp(-\mathbf{L} T_k)$ will reduce the norm for all vectors unless it is of the form $\gamma \mathbf{1}_N$. This concludes the proof of convergence.

\textbf{Case 2: } When $\frac{\rho_{\max}}{\lambda_1+\lambda_2} \leq 1$, this implies that $X^\dagger_n = \emptyset$ for all $n$ for some $k$ by definition of $\rho_{\max}$. Additionally, we have $\bar{X}_{\max} \subset \bar{X}_n$ for all $n \in \mathcal{N} $. Therefore if $x_n(t_k) \in \bar{X}_{\max}$, we have $a_{2,n}^\star (k)=a_{1,n}^\star (k)=0,\ \forall n$. This means that any $\gamma \textbf{1}_N$, $\gamma \in \overline{X}_{\max}$ is an equilibrium for the dynamics \eqref{eq:opdyn}.
\end{proof}

This theorem implies that if $\rho_{\max} > \lambda_1+ \lambda_2$, repeatedly applying the strategy $\sigma^\star$ will result in the dynamics \eqref{eq:opdyn} having a unique asymptotically stable equilibrium. The next section provides a coopetition strategy which exploits this behavior to improve the long-term utilities of both marketers simultaneously.

\section{Proposed coopetition strategy}
\label{sec:new-strategy}

Here, we use the notion of \emph{coopetition} to indicate that although the marketers compete for the market of consumers, they may have an interest in cooperating to a certain degree. And the effect is that they both may have a better long-term utility. We will refer to the underlying feature for a coopetition strategy profile candidate as sustainability. Sustainability is defined with respect to the performance obtained when Marketer $i$, $i\in \{1,2\}$, uses the strategy  $\sigma_i^{\star}$. A coopetition strategy (CS) profile is thus said to be sustainable if it Pareto-dominates the strategy profile associated with the one-shot game Nash equilibrium actions. The main purpose of this section is to propose a possible coopetition plan and prove that it is sustainable under reasonable sufficient conditions (which are met in the typical numerical setting of Sec. \ref{sec:num}). 

\begin{definition}[Sustainability] The coopetition strategy profile $(\sigma_1^{\text{CS}}, \sigma_2^{\text{CS}})$ is said to be sustainable if

\begin{equation}
\forall i \in \{1,2\}, \ U_i(  \sigma_1^{\text{CS}}, \sigma_2^{\text{CS}} ) \geq U_i(\sigma_1^{\star}, \sigma_2^{\star}). \label{eq:sust}
\end{equation} 

\end{definition}
Mathematically, the notion of sustainability corresponds to the notion of Pareto-dominance applied to two points of interest. Here, we use the more precise term sustainable to indicate that the players would accept to implement a given coopetition plan if they obtain a better utility than by using a purely competitive strategy. The proposed coopetition plan comprises two phases, the first phase is composed of all the stages $k \in \{1,2,\dots,K_1\}$ and the second phase lasts for the remaining duration, i.e. $k \in \{K_1+1,\dots,K\}$. During the first phase, both marketers repeatedly play the one-shot NE. Then, the players switch to a non-aggressive operating point such that no marketing is performed. This is held for the duration of the second phase. The rationale behind the proposed plan is that, since we have proved that the market shares stabilize according to the ratio $\eta$, the return of investment of advertising becomes negligible. The proposed coopetition plan implies that for all $i \in \{1,2\}$, $a_i(k)=a_i^\star(k)$  for all $k \in \{1,2,\dots,K_1\}$ and $a_i(k)=\mathbf{0}$ for all $k \in \{K_1+1,\dots,K\}$. Here $a_i^\star(k)$ is the action at the NE of the one-shot game as given by Lemma \ref{prop:OSNE}. The proposed coopetition plan is sustainable if both players improve their utility with respect to repeatedly playing the NE of the one-shot game. 

Here $x(t_k)$ evolves according to \eqref{eq:opdyn} with $a_i(k)=a_i^\star(k)$ for all $k$, and $\tilde{x}(t_k)$ evolves with $a_i(k)=a_i^\star(k)$ for all $k \leq K_1$ and $a_i(k)=\mathbf{0}$ for all $k > K_1$. If we have that $x_n(t_k)$ converges to some point by $K_1$ stages, such that $x(t_{k}^+)=x(t_k)$ for all $k >K_1$, this implies that the utilities for both players unilaterally improve by playing action $\mathbf{0}$. However, this condition is very conservative as it requires all the agents to converge to some opinion within a finite time. The following proposition gives a more relaxed condition for checking the feasibility of the proposed CS. When $\rho_{\max} \leq \lambda_1 +  \lambda_2$, the actions of the marketers if $x_n(t_k) \in \bar{X}_{\max}$ for all $n$ are going to be $0$ and multiple equilibria may be reached. This implies that even the strategy $\sigma^\star$ results in a non-aggressive behavior. The following proposition provides a condition under which the proposed CS is sustainable for the other case. 

\begin{proposition}
When $\rho_{\max}>\lambda_1 +  \lambda_2$, the sustainability condition \eqref{eq:sust} of the proposed CS is satisfied if $\exists \delta \in [0,1)$ such that:
\begin{itemize}
\item $x_n(t_{K_1+1}) \in [\eta-\delta,\eta+\delta]$ for all $n \in \mathcal{V}$,
\item $\delta <  \min\{ \eta,1-\eta\} \left( \frac{ \rho_{\max}}{\lambda_1+ \lambda_2}  -1\right) $,
\item \begin{equation}
\delta \leq \lambda_i \frac{ \rho_{\max} \lambda_i }{2 \rho(k)^T \mathbf{1}_N (\lambda_1+\lambda_2)^2} \left(1+ \frac{\lambda_i}{2 \rho(k)^T \mathbf{1}_N } \right)^{-1}    ,
\end{equation}
for all $k \in \{K_1+1,\dots,K\}$.
\end{itemize}
\end{proposition}
\begin{proof}
We can rewrite the condition \eqref{eq:sust} in the following manner. If $\exists K_1 \in \{0,1,\dots,K-1\}$ such that 

\begin{equation}
\begin{array}{r}
\sum_{k=1}^{K_1} u_i(\tilde{x}(t_k),a_1^\star(k),a_2^\star(k)) + \sum_{k=K_1+1}^{K} u_i(\tilde{x}(t_k),\mathbf{0},\mathbf{0}) \\
\geq \sum_{k=1}^{K} u_i(x(t_k),a_1^\star(k),a_2^\star(k))
\end{array} \label{eq:feas}
\end{equation}

for all $i \in \{1,2\}$, then the CS is sustainable. 

Since both policies play the NE for the first $K_1$ stages, the utility difference only arises from the remaining stages and can be calculated as

\begin{equation}\begin{array}{l}
\sum_{k=K_1+1}^{K} u_1(\tilde{x}(t_k),\mathbf{0},\mathbf{0}) - u_i(x(t_k),a_1^\star(k),a_2^\star(k))\\
= \sum_{k=K_1+1}^{K} \rho(k) (\tilde{x}(t_k) - x(t_k)) + \lambda_1 a_1^\star(k)^T \mathbf{1}_N .
\end{array} \label{eq:utildiff1}
\end{equation}

Note that until $t_{K_1+1}$ both use the same actions and so we have $\tilde{x}(t_{K_1+1})= x(t_{K_1+1})$. Following which, we have $\tilde{x}(t_k)= \exp(-L (t_k-t_{K_1+1})$. Due the structure of $L$, we have 

$$||\tilde{x}(t_{k+1})- x^* \mathbf{1}_N||_\infty < ||\tilde{x}(t_k)- x^* \mathbf{1}_N||_\infty$$ for all $k \in \{K_1+1,\dots,K-1\}$. Since $x(t_k) \in [x^*-\delta,x^*+\delta]$, each component of $\tilde{x}(t_k)$ is lower bounded by $x^* - \delta$.

For the dynamics of $x(t_k)$, we have the condition that $\max_{n \in \mathcal{V}} \{|x_n(t_{k+1}) - x^*|\} \leq \max_{n \in \mathcal{V}} \{|x_n(t_{k}) - x^*|\}$ for all $k \in \{K_1+1,\dots,K-1\}$. This implies that each component of $x(t_k)$ is upper bounded by $x^*+\delta$. Therefore we can lower bound the term $ \sum_{k=K_1+1}^{K} \rho (\tilde{x}(t_k) - x(t_k))$ in \eqref{eq:utildiff1} with $2\delta\sum_{k=K_1+1}^{K} \rho(k)^T \mathbf{1}_N  $ and we have

\begin{equation}\begin{array}{l}
\sum_{k=K_1+1}^{K} u_1(\tilde{x}(t_k),\mathbf{0},\mathbf{0}) - u_i(x(t_k),a_1^\star(k),a_2^\star(k))\\
\geq  -2 \delta \sum_{k=K_1+1}^{K} \rho(k)^T \mathbf{1}_N   + \lambda_1 a_1^\star(k)^T \mathbf{1}_N. 
\end{array} \label{eq:utildiff2}
\end{equation}

If this value is greater than $0$, the cooperation plan is feasible by definition. While the actions associated to the other agents may be $0$, we always have $\rho_{\max} > \lambda_1+\lambda_2$. The action at any stage $k >K_1$ is non-zero for at least one agent $m$ with $\rho_m(k) \geq \rho_{\max}$ and is given by $a_{i,m}^\star(k) \geq \frac{\rho_m \lambda_{-i} \rho_{\max}}{(\lambda_1+\lambda_2)^2} - x_m(t_k)$. From Theorem 1, we know that $\|x_n(t)-\eta\|_\infty$ is strictly decreasing. Therefore if $x_n(T_{K_1})-\eta \leq \delta$ for all $n$, then $x_n(T_{K_1})-\eta \leq \delta$ for all $n,k>K_1$. Similar arguments can be used for $U_2$. Therefore, we have the condition of sustainability to be satisfied if

 \begin{equation}
\delta \leq \lambda_i \frac{ \frac{\rho_{\max} \lambda_i}{(\lambda_1+\lambda_2)^2} - \delta  }{2 \rho(k)^T \mathbf{1}_N}     ,
\end{equation}

since we have at least one agent $m$ with $a_{i,m}^\star$ lower bounded by $\frac{\rho_{\max} \lambda_i}{(\lambda_1+\lambda_2)^2} - \delta$. Then $x_m(t_k) \in X_m^\dagger$ and $\rho_m(k)\geq  \rho_{\max}$. 
\end{proof}

Next, we provide a numerical example which illustrates the sustainability of the proposed CS and allow us to assess the benefits of coopetition in the long-term.

\section{Numerical performance analysis}
\label{sec:num}

To conduct a good comparison analysis, we choose values for the key parameters that are typical and quite similar to  \cite{varma2018marketing}.  For the costs of advertising, we assume that: $\lambda_1=1$ and $\lambda_2=0.5$. We consider a cascading graph structure where the a sub-graph structure of $5$ agents is repeated. The set of edges defining the sub-graph are given by $\mathcal{E}=\{(1,5),(2,1),(2,3),(3,1),(3,5),(4,1),(5,1),(5,2)\}$ with the connection weight fixed at $1$ when the edge exists. Moreover, we connect the repeating blocks of 5 agents in the following manner: agent $n$ is connected to $n+5$ (when $n<N-5$) with $\mathbf{L}_{n,n+5}=-1$ and $\mathbf{L}_{n+5,n}=-4$, i.e., the preceding blocks are more influential. The initial opinions are taken to be $x_n(0)=0.4+ \frac{n}{2N}$ and we calculate $\rho=\mathbf{1}_N^{\top} \exp(-\mathbf{L})$ where $\mathbf{L}$ is the Laplacian of the resulting graph. We also consider $T_k=1$, $K=5$ and $t_k=k$.

Fig.~\ref{simfig1} represents the evolution of the opinions when the number of agents $N=50$ and both marketers implement the one-shot NE for all campaign stages. We plot the opinions of agents $n=1$, $n=15$ and $n=50$ to show the types of behavior observable. The agent $n=1$ has a high influential power and is therefore controlled by both marketers while the agent $n=50$ is uncontrolled and slowly converges to $\eta$ by following his neighbors. The agent $n=15$ is controlled only when its opinion is far from $\eta$. This figure clearly shows one important result which is missing in the analysis conducted in  \cite{varma2018marketing}. Despite the presence of the zero-sum component in the stage game that creates tension in the network, the state of the network stabilizes to a given value which can be predicted from the theoretical analysis. Here, this value corresponds to $\eta$ and equals $\frac{1}{3}$. We observe that $\|x(t_k)-\eta \|_{\infty}<0.01$ by $k \geq 5$.

\begin{figure}[th!]
\begin{center}
\includegraphics[width=8cm,trim={1cm 0 0 1cm},clip]{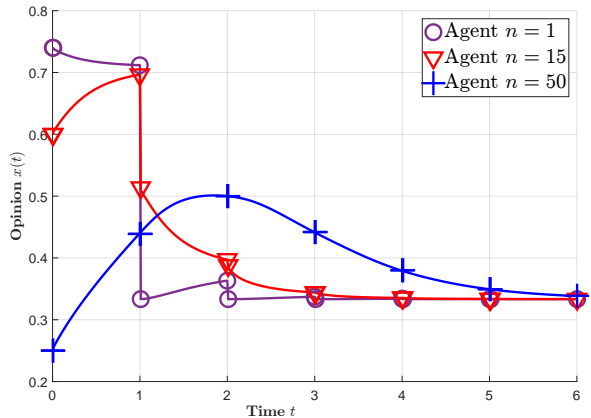}
\caption{Although the marketing game has a zero-sum component which creates tension in the network, the network reaches a consensus after sufficient stages.}\label{simfig1}
\end{center}\end{figure}

Fig.~\ref{simfig3}, represents the long-term utilities for the two marketers (with the parameters as before) as a function of $K_1$ (namely, the number of stages of the first phase of the proposed coopetition plan). It is seen that for any $1 \leq K_1 <5$, both marketers obtain a better long-term utility by stopping their marketing after $K_1$ campaigns. One of the virtues of this observation (that illustrates Theorem 1) is to show the potential of designing long-term marketing strategies and thus using a dynamic game formulation instead of exploiting a static game model as in \cite{varma2018marketing}. In Table 1, we compare the stage utilities by playing the proposed strategy in comparison to the one-shot NE after convergence to $\eta$. 

\begin{figure}[th!]
\begin{center}
\includegraphics[width=8cm, trim={1cm 0 0 1cm},clip]{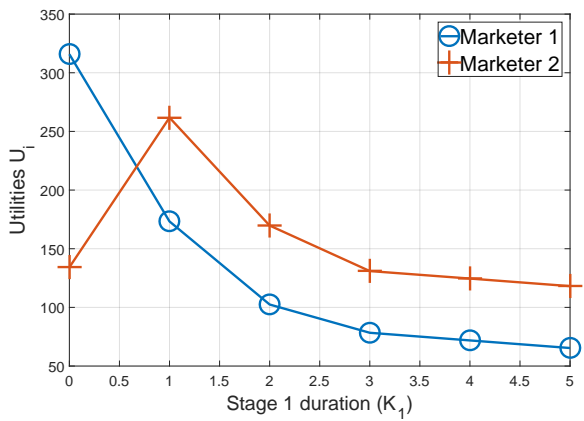}
\caption{The proposed coopetition strategy profile is seen to be sustainable for $K_1\geq 1$ that is, it Pareto-dominates the strategy profile of \cite{varma2018marketing}. }\label{simfig3}
\end{center}\end{figure}
 
\vspace{3mm}

\begin{table}[h!]
\vspace{3mm}
 \centering
\begin{tabular}{|c|c|c|c|c|}
\hline
   $N$    &  50   & 100    & 200\\
  \hline
Proposed marketing strategy ($i=1$)  & 17 & 33 & 67\\
\hline
Proposed marketing strategy ($i=2$)  & 34 & 66 & 132 \\
\hline
Strategy of \cite{varma2018marketing} ($i=1$) & 13 &28& 58 \\
\hline
Strategy of \cite{varma2018marketing} ($i=2$)  & 30 & 61 & 124 \\
\hline
Stages required for convergence to $\eta$ & 5 & 6 & 6 \\
\hline
\end{tabular}\vspace{2mm}
\caption{Stage utilities with proposed marketing strategy compared to strategy in \cite{varma2018marketing} after practical convergence.}
\end{table}

\section{Conclusion}
\label{sec:con}
In this paper, we study a game model which characterizes the repeated competition between firms trying to capture a market share by advertising over social media. The consumers that interact over the social network are therefore not only under the influence of the other consumers of the network but also of the external marketers who influence them through campaigns. This leads to a hybrid dynamics of consumers' opinions. Exploiting the key results in \cite{varma2018marketing}, we propose a coopetition marketing strategy which combines the one-shot Nash equilibrium actions and no advertising. Under reasonable sufficient conditions, it is proved that the proposed coopetition strategy profile Pareto-dominates the solution of \cite{varma2018marketing}. Numerical examples illustrate the theoretical results.

\bibliography{RG19Final}

\end{document}